\documentclass[12pt,english]{article}
\usepackage{a4}
\usepackage[round]{natbib}
\usepackage{babel}

\usepackage[margin=1.15in]{geometry}
\usepackage{sgame}
\usepackage{amsmath,amssymb,bm}
\usepackage{amsthm}
\usepackage{mathtools}
\usepackage{amsbsy}
\usepackage{caption}
\usepackage{extarrows}
\usepackage[colorlinks = true,linkcolor = purple,urlcolor  = blue,citecolor = blue,anchorcolor = blue]{hyperref}
\usepackage[multiple]{footmisc}
\usepackage{setspace}
\usepackage{xfrac}
\usepackage{float}
\usepackage{tikz}
\usetikzlibrary{positioning}
\usepackage[affil-it]{authblk}
\usepackage{mathptmx} 

\begin{document}

\title{A choice-based axiomatization of Nash equilibrium\footnote{Many thanks to Tiago C. Botelho, Peter Streufert, and Mark Voorneveld for valuable input, and to reviewers from the LOFT 16 Conference (London, 2026) for helpful comments. I am especially grateful to the Advisory Editor and two reviewers from Games and Economic Behavior for insightful comments that helped me greatly to improve the paper. Financial support from the OP Pohjola Research Foundation (Grant Nos. 20240089 and 20250085) is gratefully acknowledged.}}
\author{Michele Crescenzi\\ \href{mailto:michele.crescenzi@helsinki.fi}{michele.crescenzi@helsinki.fi}}
\affil{University of Helsinki and Helsinki Graduate School of Economics}

\newtheorem{proposition}{Proposition}
\newtheorem{theorem}{Theorem}
\newtheorem*{theorem*}{Theorem}
\newtheorem*{definition*}{Definition}
\newtheorem{assumption}{Assumption}
\newtheorem{definition}{Definition}
\newtheorem{lemma}{Lemma}
\newtheorem{claim}{Claim}
\newtheorem{remark}{Remark}
\newtheorem*{remark*}{Remark}
\newtheorem{corollary}{Corollary}
\theoremstyle{definition}
\newtheorem{example}{Example}
\newtheorem*{example*}{Example}

\maketitle

\begin{abstract}
An axiomatization of Nash equilibrium is provided for games of complete information in normal form. The Nash equilibrium correspondence is shown to be the unique solution concept that satisfies four simple and intuitive axioms. Two axioms reflect a mild form of individual rationality: one states that any profile of dominant strategies in a game must be a solution, whereas the other says that any solution to a game should not depend on strictly dominated strategies. The remaining two axioms are adaptations to non-cooperative games of standard choice-theoretic axioms about contraction and expansion consistency.

\bigskip

JEL CLASSIFICATION: C72
\bigskip

\end{abstract}

\newpage

\onehalfspacing

\section{Introduction}
Axiomatizations bring out the essential properties of solution concepts. Predominant in cooperative game theory, the axiomatic method has been employed to study solution concepts for non-cooperative games too, especially Nash equilibrium and its refinements. Most recently, the axiomatic analysis of Nash equilibrium has been revitalized by the work of \cite{voorneveld2019axioms}, \cite{brandl2024axioms} and \cite{Sandomirskiy-et-al2025}. The present paper contributes to this active research stream by offering a new axiomatization of the Nash equilibrium correspondence for normal-form games of complete information with finitely many players and strategy sets of any cardinality.

A non-cooperative solution concept can be viewed as the outcome of a choice problem in which a hypothetical and impartial selector chooses a set of strategy profiles for each game within a certain domain, i.e., a well-defined class of games. A solution concept is thus determined by the principles, formalized as axioms, which inform the selector's choices. I show in this paper that Nash equilibrium is fully characterized by four simple and intuitive axioms, two of which presuppose a mild form of individual rationality. The remaining two axioms, instead, are adaptations of standard choice-theoretic axioms about contraction and expansion consistency, thus making explicit the connection between Nash equilibrium and choice theory.

Here is an informal description of the four axioms used in this paper. The first is called \textit{joint dominance} (JD). It says that if a strategy profile exists in which each player picks a dominant strategy--i.e., a strategy that maximizes the player's preferences irrespective of what her opponents are going to choose--then this profile is a solution to the game. JD reflects a very mild form of individual rationality, in the sense that a solution concept must select all profiles of strategies that are undoubtedly optimal for each player.

The second axiom, \textit{invariance to strictly dominated strategies} (ISDS), postulates that whenever a game $G'$ can be formed by removing only strictly dominated strategies from another game $G$, then $G$ and $G'$ have the very same set of solutions. ISDS expresses the principle that a solution to a game should not depend on strictly dominated strategies, as these are never played by rational players. 

As for the third axiom, suppose a game $G'$ is obtained out of another game $G$ by removing some strategies of one or more players. The \textit{independence of irrelevant strategies} (IIS) says that if a strategy profile belonging to both $G$ and $G'$ is a solution to $G$, then that strategy profile is a solution to the ``smaller'' game $G'$ too. IIS is a contraction consistency property and is an adaptation to non-cooperative games of Sen's property $\alpha$ in choice theory \citep{sen1971choice}, which in turn can be traced back to Postulate 4 in \cite{chernoff1954rational}.

To illustrate the fourth axiom, consider three games having the same set of two players and with the following payoff matrices, in which $w$, $x$, $y$, and $z$ are arbitrary real numbers.

\begin{figure}[htbp]
	\centering
	\begin{minipage}[b]{0.3\textwidth}
		\centering
		\begin{game}{2}{2}
			&  $A$ &  $B$ \\
			$D$        & $2,5$ & $3,1$\\
			$E$       & $1,0$ & $4,1$ 
		\end{game}
		\caption*{$G^1$}
	\end{minipage}
	\hfill
	\begin{minipage}[b]{0.3\textwidth}
		\centering
		\begin{game}{2}{2}
			&  $B$ &  $C$ \\
			$E$        & $4,1$ & $2,0$\\
			$F$       & $1,2$ & $4,1$ 
		\end{game}
		\caption*{$G^2$}
	\end{minipage}
	\hfill
	\begin{minipage}[b]{0.3\textwidth}
		\centering
	\begin{game}{3}{3}
		&  $A$ &  $B$ &  $C$\\
		$D$         & $2,5$ & $3,1$ & $x,y$\\
		$E$       & $1,0$ & $4,1$ & $2,0$ \\
		$F$   & $w,z$   & $1,2$ & $4,1$
	\end{game}
		\caption*{$G^3$}
	\end{minipage}
\end{figure}

There is a sense in which $G^3$ can be formed by ``merging'' $G^1$ and $G^2$: each player's set of strategies in $G^3$ is the union of her strategy sets from $G^1$ and $G^2$; further, each player's payoff function in $G^1$ and $G^2$ is the restriction of her payoff function in $G^3$ to $G^1$ and $G^2$, respectively. Now suppose that the strategy profile $(E,B)$ solves both $G^1$ and $G^2$. The \textit{merging consistency} (MC) axiom posits that $(E,B)$ must solve $G^3$ as well. That is, a strategy profile which solves two distinct games must also solve any game formed by merging those two games. MC is thus an expansion consistency property which mirrors Sen's property $\gamma$ \citep{sen1971choice}, which in turn can be traced back to Postulate 10 in \cite{chernoff1954rational}.

Although the IIS axiom is already used by \cite{salonen1992axioms}, \cite{peleg-tijs1996consistency} and \cite{forgo2015axiomsNash} in their axiomatic analyses of Nash equilibrium, all of which are discussed below, I could not find any previous axiomatization of Nash equilibrium that adopts any one of the other three axioms. On a related note, \cite{sprumont2000Nash} shows that the profiles of choices made by a group of individuals are Nash-rationalizable if and only if two conditions are met which are just about the same as the IIS and MC axioms studied here. Nonetheless, as I argue in more detail below, \cite{sprumont2000Nash} studies choices made by individuals whose preferences are not observable, whereas in this paper I study solution concepts under the assumption that players' preferences are publicly observable.

From a methodological perspective, most axioms are conditional statements whose antecedents apply to finite collections of games. With the exception of JO, whose antecedent refers to a single game, the axioms employed in this paper are defined on games having the same set of players and differing only in their sets of strategy profiles.\footnote{Players' preferences do not differ either, in the sense that their restrictions to overlapping strategy profiles coincide.} Specifically, IIS and ISDS are defined on pairs of games with nested sets of strategy profiles, while MC applies to any game that one can form by merging the sets of strategy profiles of two constituent games. This is in stark contrast to many axioms in the existing literature, most notably the consistency principle \citep{peleg-tijs1996consistency}, which applies to games with nested sets of players, as well as axioms defined on games differing only in their payoff functions, a recent example of which is the independence of identical consequences \citep{voorneveld2019axioms}.

\paragraph{Related work.}
The consistency principle \citep{peleg-tijs1996consistency} is possibly the central axiom in the early axiomatic studies of Nash equilibrium, which are surveyed in \citet[Section 12]{thomson2001axiomatic-method}. Consistency plays the same role as IIS does in this paper. Nonetheless, as I show in Section \ref{sub:OtherAxioms}, the two axioms are logically independent since they apply to different families of games: consistency is defined on games with nested sets of players, whereas IIS is defined on games having the very same set of players. \cite{peleg-tijs1996consistency} also offer an alternative axiomatization of Nash equilibrium which is based on IIS (instead of consistency) and \textit{converse} consistency. The latter is analogous to MC and, just like consistency, applies to games with nested sets of players. I show in Section \ref{sub:OtherAxioms} that converse consistency and MC are logically independent.

The IIS axiom is also used by \cite{salonen1992axioms}. His main result states that any solution concept satisfying IIS and another property (called the weak principle of fair compromise) is a weak refinement of the Nash equilibrium correspondence. His work is restricted to games with continuous and quasi-concave payoff functions and to solution concepts that select a non-empty set of strategy profiles for every game. By contrast, in this paper I do not put any of these restrictions and provide a full characterization of the Nash equilibrium correspondence.

\cite{forgo2015axiomsNash} shares many similarities with this paper in that his axioms, one of which is IIS, are inspired by choice theory and apply to games having the same set of players. His main result is that in finite games having at least one Nash equilibrium in pure strategies, Nash equilibrium is the unique solution concept that satisfies non-emptiness, IIS, converse independence of irrelevant strategies, and a weak form of one-person rationality. There are three essential differences from my analysis, though. First, as I show in Section \ref{sub:OtherAxioms}, the converse independence of irrelevant strategies of \cite{forgo2015axiomsNash} is implied by the MC axiom introduced in this paper but the converse is not true. Second, while \cite{forgo2015axiomsNash} restricts his analysis to games in which every player has finitely many strategies, I do not put any constraint on the cardinality of strategy sets. Third, his axiomatization applies only to games having at least one Nash equilibrium in pure strategies, whereas this paper's axiomatization is strictly more general and covers Nash equilibria in mixed strategies as well as games in which Nash equilibria do not exist at all.

The most recent axiomatic analyses of Nash equilibrium are offered by \cite{voorneveld2019axioms}, \cite{brandl2024axioms} and \cite{Sandomirskiy-et-al2025}. There are fundamental differences between these papers and mine. For one thing, I cover solution concepts that may select an empty set of strategy profiles and I study axioms that apply to pure and mixed strategies alike. Conversely, \cite{brandl2024axioms} and \cite{Sandomirskiy-et-al2025} focus on solution concepts that select a non-empty set of mixed strategy profiles for every game. For another thing, the IIS, MC and ISDS axioms introduced here are defined for games whose sets of strategy profiles are nested. This contrasts with the independence of identical consequences of \cite{voorneveld2019axioms} and the consistency axiom of \cite{brandl2024axioms}, both of which apply to games that differ in their payoff functions but not in their sets of strategy profiles.

In recent work, \cite{govindan2026} carry out an axiomatic analysis of solution concepts that are refinements of Nash equilibrium. My paper contributes marginally to this strand of the literature by showing in Lemma \ref{lemma:SolNash} that any strict refinement of Nash equilibrium must violate at least one of MC and JD.

Finally, \cite{sprumont2000Nash} shows that a joint choice function is Nash-rationalizable\footnote{In a related work, \cite{ray-zhou2001} provide necessary and sufficient condition for the rationalization in subgame-perfect equilibrium of joint choice functions.} if and only if it satisfies two properties which are just about the same as the IIS and MC axioms. Specifically, MC is essentially the same as the persistence under expansion of \cite{sprumont2000Nash}, whereas IIS is, \emph{mutatis mutandis}, equivalent to the first part of his persistence under contraction. The key difference between his study and mine is the object of analysis. \cite{sprumont2000Nash} studies joint choice functions, which are choice functions over Cartesian products of sets of actions. Joint choice functions rest on the assumption that individual preferences are not observable. By contrast, I study solution concepts which are defined over standard games and rely on the assumption that players' preferences are publicly observable. Notwithstanding these differences, one can view the present paper as an extension of Sprumont's analysis to solution concepts.

\paragraph{Paper structure.}
Preliminary definitions are given in Section \ref{sec:Prelim}. Axioms and the axiomatic characterization of Nash equilibrium are in Section \ref{sec:Ax}. Section \ref{sec:Disc} concludes with a discussion of the logical independence of the four axioms, the special case of one-player games, the domain of the axiomatization, and a comparison with other axioms from the existing literature.

\section{Preliminary definitions}\label{sec:Prelim}
A game in normal form is a list  $G = \langle I, \left(S_i\right)_{i\in I}, \left(\succeq_i\right)_{i\in I}\rangle$, where $I= \left\lbrace 1, \dots, n\right\rbrace$ is a set of $n\geq 1$ players; for each player $i\in I$, the set $S_i$ is a non-empty set of pure or mixed\footnote{Throughout this paper, it is implicitly assumed that players treat mixed strategies as objects of choice.} strategies and $\succeq_i$ is a complete and transitive preference relation over the set of strategy profiles $S = S_1 \times S_2 \times \dots \times S_n$. The strict part of a preference relation $\succeq_i$ is denoted by $\succ_i$ while $\sim_i$ denotes its symmetric part. Preferences need not have a utility or expected utility representation.

For any $i\in I$, a strategy profile $s = (s_1, s_2, \dots, s_n)$ can also be written as $(s_i, s_{-i})$, where $s_i$ is the $i$th element of $s$ while $s_{-i}$ is the list of all other elements of $s$. In addition, $S_{-i}$ stands for the set $\times_{j\in I\setminus \{i\}} S_j$. It is understood that $(s_i, s_{-i})$ reduces to $s_i$ when $n = 1$. It is equally understood that in the definitions below, the expressions ``for all $s_ {-i} \in S_{-i}$'' and ``for all $t_ {-i} \in S_{-i}$'' should be ignored when $n = 1$.

Given two games $G = \langle I, \left(S_i\right)_{i\in I}, \left(\succeq_i\right)_{i\in I}\rangle$ and $G' = \langle I, \left(S'_i\right)_{i\in I}, \left(\succeq'_i\right)_{i\in I}\rangle$ having the same set of players $I$, we say that $G'$ is a \textbf{reduction}\footnote{Reductions arise naturally in iterated elimination procedures--see, e.g., \cite{pearce1984rationalizable}.} of $G$ if, for all $i\in I$, we have $\emptyset \neq S_i' \subseteq S_i$ and the preference relation $\succeq'_i$ is the restriction of $\succeq_i$ to $S'$, that is, $\succeq'_i \; = \; \succeq_i \cap \left(S' \times S' \right)$. If, in addition, $\vert S'_j \vert = 1$ for some $j\in I$ then we say that the reduction $G'$ has a \textbf{dummy}; similarly, $G'$ is said to have a \textbf{doubleton strategy set} if $\vert S_j' \vert = 2$ for some $j\in I$.

The game $G'= \langle I, \left(S'_i\right)_{i\in I}, \left(\succeq'_i\right)_{i\in I}\rangle$ is called a \textbf{strict reduction} of $G = \langle I, \left(S_i\right)_{i\in I}, \left(\succeq_i\right)_{i\in I}\rangle$ if $G'$ is a reduction of $G$ satisfying the following conditions:
\begin{itemize}
\item there is a player $j\in I$ for which $S_j \setminus S'_j \neq \emptyset$;
\item for all players $i\in I$, if $s_i \in S_i \setminus S'_i$, then there exists a strategy $s'_i\in S'_i$ such that
\begin{equation*}
(s'_i, s_{-i}) \succ_i (s_i, s_{-i}) \text{ for all } s_{-i}\in S_{-i}. 
\end{equation*}
\end{itemize}
In words, a strict reduction $G'$ is obtained by removing only some strictly dominated strategies from $G$. Any game has at least one reduction (namely, the game itself) while it may not have any strict reduction.

Given two Cartesian products $S = S_1 \times S_2 \times \dots \times S_n$ and $T = T_1 \times T_2 \times \dots \times T_n$, let
\begin{equation}\label{eq:CartUnion}
	S \vee T := \left(S_1 \cup T_1 \right) \times \left(S_2 \cup T_2 \right) \times \dots \times \left(S_n \cup T_n \right).
\end{equation}
In words, the set $S \vee T$ is the smallest Cartesian product that contains $S \cup T$.

Let $\Gamma$ be a non-empty class of games. A \textbf{solution concept} on $\Gamma$ is a mapping $\varphi$ that assigns to each game $G = \langle I, \left(S_i\right)_{i\in I}, \left(\succeq_i\right)_{i\in I}\rangle$ in $\Gamma$ a (possibly empty) set of strategy profiles $\varphi(G)$ such that $\varphi(G) \subseteq S$.

A \textbf{Nash equilibrium} of a game $G = \langle I, \left(S_i\right)_{i\in I}, \left(\succeq_i\right)_{i\in I}\rangle$ is a strategy profile $s$ such that, for all $i\in I$,
\begin{equation*}
	(s_i, s_{-i}) \succeq_i (t_i, s_{-i}) \text{ for all } t_i\in S_i.
\end{equation*}
The solution concept that assigns to each game its (possibly empty) set of all Nash equilibria is denoted by $\varphi^{\mathsf{NE}}$ and is often called \textit{Nash equilibrium correspondence}.

Finally, a non-empty class of games $\Gamma$ is \textit{closed with respect to dummies and doubleton strategy sets} (in short: \textbf{d-closed}) if, for each game $G$ in $\Gamma$, the class $\Gamma$ contains all reductions of $G$ that have a dummy or a doubleton strategy set.

\section{Axiomatic analysis}\label{sec:Ax}
\subsection{Axioms}\label{sub:Axioms}
Each of the following axioms is defined on a non-empty class of games $\Gamma$.

\paragraph*{Joint dominance (JD).} \emph{If $G = \langle I, \left(S_i\right)_{i\in I}, \left(\succeq_i\right)_{i\in I}\rangle$ is in $\Gamma$ and $s \in S$ is a strategy profile such that, for all $i\in I$,}  
\begin{equation*}
	(s_i, t_{-i}) \succeq_i (t_i, t_{-i}) \textit{ for all } t_i \in S_i \textit{ and all } t_{-i} \in S_{-i},
\end{equation*}
\emph{then $s \in \varphi(G)$.}

That is, every profile of dominant strategies\footnote{I define \textit{dominant strategies} as \cite{gibbard1973manipulation} and subsequent papers on dominant strategy implementation do. Note that it is not necessarily the case that a dominant strategy $s_i \in S_i$ weakly dominates all the other strategies in $S_i$: a strategy $s_i' \in S_i$ could exist such that $s_i \neq s_i'$ and $(s_i, t_{-i}) \sim_i (s'_i, t_{-i})$ for all $t_{-i} \in S_{-i}$.} is a solution.\footnote{JD resembles the unanimity axiom of \cite{voorneveld2019axioms}, but the two are logically independent.} This axiom blends individual rationality and strategic robustness: a solution concept satisfying JD cannot discard strategy profiles in which each strategy maximizes a player's preference relation regardless of how her opponents will play the game.

\paragraph*{Invariance to strictly dominated strategies (ISDS).} \emph{If $G$ and $G'$ are in $\Gamma$ and $G'$ is a strict reduction of $G$, then $\varphi(G) = \varphi(G')$.}

ISDS is an invariance property based on a mild form of rationality. It says that the solution set of a game is left unchanged by the removal of strictly dominated strategies. ISDS bears some resemblance to Postulate 5 in \cite{chernoff1954rational}, which formalizes invariance to weakly dominated strategies in one-person decision problems under uncertainty.
	
\paragraph*{Independence of irrelevant strategies (IIS).} \emph{If $G = \langle I, \left(S_i\right)_{i\in I}, \left(\succeq_i\right)_{i\in I}\rangle$ and $G' = \langle I, \left(S'_i\right)_{i\in I}, \left(\succeq'_i\right)_{i\in I}\rangle$ are in $\Gamma$ and $G'$ is a reduction of $G$, then}
\begin{equation*}
	S' \cap \varphi(G) \subseteq \varphi(G').
\end{equation*}
That is, if a strategy profile is a solution to a game $G$, then that strategy profile must be a solution to every reduction of $G$ to which it belongs. IIS is a contraction consistency property: it is the natural adaptation to strategic games of Postulate 4 in \cite{chernoff1954rational}, which is also known as Sen's property $\alpha$\footnote{A choice correspondence $c$ on a set $X$ satisfies property $\alpha$ if for all $x \in X$ and all non-empty $A,B \subseteq X$, if $x \in c(A)$ and $x \in B \subseteq A$ then $x\in c(B)$. It is well known that Sen's property $\alpha$ is equivalent to \textit{weak path independence} $\alpha$ \citep[Theorem 2.11]{suzumura1983book}, which states that $c(A \cup B) \subseteq c(c(A) \cup c(B))$ for all non-empty $A,B \subseteq X$.} in choice theory \citep{sen1971choice}. IIS is used in the axiomatic analysis of Nash equilibrium by \cite{salonen1992axioms}, \cite{peleg-tijs1996consistency} and \cite{forgo2015axiomsNash}.

\paragraph*{Merging consistency (MC).} \emph{For any game $G = \langle I, \left(S_i\right)_{i\in I}, \left(\succeq_i\right)_{i\in I}\rangle$ in $\Gamma$, if two of its reductions $G' = \langle I, \left(S'_i\right)_{i\in I}, \left(\succeq'_i\right)_{i\in I}\rangle$ and $G'' = \langle I, \left(S''_i\right)_{i\in I}, \left(\succeq''_i\right)_{i\in I}\rangle$ are in $\Gamma$ and are such that $S' \vee S'' = S$, then}
\begin{equation*}
\varphi (G') \cap \varphi (G'') \subseteq \varphi (G).
\end{equation*}

MC is about expansion consistency. It says that if a game $G$ can be formed by merging two of its reductions $G'$ and $G''$, then any strategy profile that solves both $G'$ and $G''$ must solve $G$ as well. As per \eqref{eq:CartUnion}, the merging of $G'$ and $G''$ consists in taking the player-by-player union of strategy sets from $G'$ and $G''$. MC is an adaptation to strategic games of Postulate 10 in \cite{chernoff1954rational} and is analogous to Sen's property $\gamma$\footnote{A choice correspondence $c$ on a set $X$ satisfies property $\gamma$ if for all non-empty families $\{A_j\}_{j\in J}$ such that $\cup_{j\in J} A_j = B \subseteq X$, it holds that $\cap_{j\in J} c(A_j) \subseteq c(B)$.} \citep{sen1971choice}.

\subsection{Characterization}\label{sec:AxChar}
The following lemma is the basis for the characterization of Nash equilibrium in Theorem \ref{thm:IFF}.

\begin{lemma}\label{lemma:SolNash}
Let $\Gamma$ be a d-closed class of games and $\varphi$ a solution concept on $\Gamma$.
\begin{itemize}
\item[(a)] If $\varphi$ satisfies IIS and ISDS, then $\varphi(G) \subseteq \varphi^{\mathsf{NE}}(G)$ for all $G$ in $\Gamma$.
\item[(b)] If $\varphi$ satisfies MC and JD, then $\varphi^{\mathsf{NE}}(G) \subseteq \varphi(G)$ for all $G$ in $\Gamma$.
\end{itemize}
\end{lemma}
\begin{proof}
To prove part (a), suppose a solution concept $\varphi$ on a d-closed class $\Gamma$ satisfies IIS and ISDS. By way of contradiction, suppose there are a game $G = \langle I, \left(S_i\right)_{i\in I}, \left(\succeq_i\right)_{i\in I}\rangle$ in $\Gamma$ and a strategy profile $s = (s_1, \dots, s_n) \in \varphi(G)$ that is not a Nash equilibrium of $G$. Thus, there must exist a player $j \in I$ and a strategy $t_j \in S_j$ such that $t_j \neq s_j$ and $(t_j, s_{-j}) \succ_j (s_j, s_{-j})$.

Now take the games $G' = \langle I, \left(S'_i\right)_{i\in I}, \left(\succeq'_i\right)_{i\in I}\rangle$ and $G'' = \langle I, \left(S''_i\right)_{i\in I}, \left(\succeq''_i\right)_{i\in I}\rangle$ for which:
\begin{itemize}
\item $S'_j = \{s_j, t_j\}$ and $S''_j = \{t_j\}$;
\item for all $i\in I\setminus{\{j\}}$, we have $S'_i = S''_i = \{s_i\}$;
\item for all $i\in I$, the preference relations $\succeq'_i$ and $\succeq''_i$ are the restrictions of $\succeq_i$ to $S'$ and $S''$, respectively.
\end{itemize}

The games $G'$ and $G''$ are both in $\Gamma$ since they are reductions of $G$ with a doubleton strategy set and a dummy, respectively. In addition, we have $s = (s_j, s_{-j})\in \varphi(G')$ by IIS. Since $G''$ is a strict reduction of $G'$, we get $\varphi(G') = \varphi(G'')$ from ISDS. But then we have $(s_j, s_{-j})\in \varphi(G'') = S''$, which contradicts $t_j \neq s_j$.

\bigskip

For part (b), suppose a solution concept $\varphi$ on a d-closed class $\Gamma$ satisfies MC and JD. Take any game $G = \langle I, \left(S_i\right)_{i\in I}, \left(\succeq_i\right)_{i\in I}\rangle$ in $\Gamma$ and let $s = (s_1, \dots, s_n) \in \varphi^{\mathsf{NE}}(G)$. To prove that $s \in \varphi(G)$, we construct $n$ games $G^1, \dots, G^n$ as follows. For each $k \in \{1, \dots, n\}$, the game $G^k = \langle I, \left(S^k_i\right)_{i\in I}, \left(\succeq^k_i\right)_{i\in I}\rangle$ is such that, for all $i\in I$, we have:
\begin{itemize}
	\item $S^k_i = S_i$ if $i\leq k$, and $S^k_i = \{s_i\}$ if $i > k$;
	\item the preference relation $\succeq^k_i$ is the restriction of $\succeq_i$ to $S^k = S^k_1 \times S^k_2 \times \dots \times S^k_n$.
\end{itemize}

Each game $G^k$ is either equal to $G$ or a reduction of $G$ with a dummy, thus belonging to $\Gamma$. It is clear that $s \in \cap_{k =1}^n \varphi^{\mathsf{NE}}(G^k)$ and $G^n = G$. We are going to show that $s \in \cap_{k =1}^n \varphi(G^k)$. First, $s \in \varphi(G^1)$ follows from the JD axiom, thus ending the proof if $n=1$. Next, suppose $s \in \varphi(G^h)$ for some $h \in \{1, \dots, n-1\}$. Consider the game $H = \langle I, \left(T_i\right)_{i\in I}, (\hat{\succeq}_i)_{i\in I}\rangle$ such that, for all $i\in I$, we have:
\begin{itemize}
	\item $T_i = S_i$ if $i = h + 1$, and $T_i = \{s_i\}$ if $i \neq h + 1$;
	\item the preference relation $\hat{\succeq}_i$ is the restriction of $\succeq_i$ to $T$.
\end{itemize}

Clearly, $H$ belongs to $\Gamma$ and $s \in \varphi^{\mathsf{NE}}(H)$. In addition, $s \in \varphi(H)$ by virtue of the JD axiom. It is also easy to verify that $G^h$ and $H$ are two reductions of $G^{h+1}$ such that $S^h \vee T = S^{h+1}$. Hence, we get $s\in \varphi(G^{h+1})$ from the MC axiom. Thus we can conclude that $s$ is a solution to all games $G^1, \dots, G^n = G$, so ending the proof.
\end{proof}

Part (a) of the lemma above says that any solution concept satisfying IIS and ISDS on a d-closed class of games must be a weak refinement of Nash equilibrium. Part (b) implies that any strict refinement of Nash equilibrium must violate at least one of MC and JD.

The central result of this paper is that the Nash equilibrium correspondence is fully characterized by the four axioms introduced in the previous subsection.

\begin{theorem}\label{thm:IFF}
Let $\Gamma$ be a d-closed class of games. A solution concept $\varphi$ on $\Gamma$ satisfies IIS, MC, ISDS and JD if and only if $\varphi (G) = \varphi^{\mathsf{NE}} (G)$ for all $G$ in $\Gamma$.
\end{theorem}
\begin{proof}
The ``only if'' part of the theorem is a corollary of Lemma \ref{lemma:SolNash}.

For the ``if'' part, it is easy to check that if $s$ is a Nash equilibrium of a game $G$, then $s$ is a Nash equilibrium of every reduction of $G$ in which $s$ is a feasible strategy profile, thus proving IIS. It is also easy to check that the set of all Nash equilibria of a game is left unchanged by the removal of any strictly dominated strategy, from which ISDS follows. In addition, JD is an immediate consequence of the definition of a Nash equilibrium. Finally, to show that MC holds, let $G = \langle I, \left(S_i\right)_{i\in I}, \left(\succeq_i\right)_{i\in I}\rangle$ be a game in $\Gamma$ and take two of its reductions $G' = \langle I, \left(S'_i\right)_{i\in I}, \left(\succeq'_i\right)_{i\in I}\rangle$ and $G'' = \langle I, \left(S''_i\right)_{i\in I}, \left(\succeq''_i\right)_{i\in I}\rangle$ in $\Gamma$ such that $S' \vee S'' = S$. Suppose that
\begin{equation*}
s = (s_1, \dots, s_n) \in \varphi^{\mathsf{NE}} (G') \cap \varphi^{\mathsf{NE}} (G'').
\end{equation*}
Then, for all $i\in I$, we have $(s_i, s_{-i}) \succeq_i' (s'_i, s_{-i})$ for all $s_i' \in S_i'$ and $(s_i, s_{-i}) \succeq_i'' (s''_i, s_{-i})$ for all $s_i'' \in S_i''$. Since $\succeq_i'$ and $\succeq_i''$ are both restrictions of $\succeq_i$, we have $(s_i, s_{-i}) \succeq_i (t_i, s_{-i})$ for all $t_i \in S_i' \cup S''_i = S_i$, from which it follows that $s \in \varphi^{\mathsf{NE}} (G)$.
\end{proof}

\section{Discussion}\label{sec:Disc}
\subsection{Logical independence}
The examples in this subsection show that the four axioms introduced in Section \ref{sub:Axioms} are logically independent. For convenience, here preference relations are represented by payoff functions: each player $i$ has a payoff function $u_i: S \to \mathbb{R}$ such that $s \succeq_i t$ if and only if $u_i(s) \geq u_i(t)$.
\begin{example}\label{ex:LogigIndep1}
Suppose $\Gamma$ is the smallest d-closed class of games that contains the canonical Prisoner's Dilemma. The solution concept that assigns the empty set to each game in $\Gamma$ satisfies IIS, MC and ISDS but not JD. On the other hand, the solution concept that assigns to each game in $\Gamma$ its entire set of strategy profiles satisfies IIS, MC and JD but not ISDS. \qed
\end{example}

\begin{example}\label{ex:LogigIndep2}
Suppose $\Gamma$ is the smallest d-closed class containing the two-by-two game $G$ with the payoff matrix below.

\begin{figure}[H]
\centering
\begin{game}{2}{2}
	& $L$    & $R$ \\
$U$      & $2,2$  & $0,0$ \\
$D$      & $1,1$  & $1,1$ 
\end{game}
\end{figure}

Consider the solution concept that assigns to each game in $\Gamma$ the set of its strong Nash equilibria.\footnote{A strong Nash equilibrium \citep{aumann1959} is a strategy profile from which no coalition of players can profitably deviate, given that the strategies of the other players remain fixed.} This solution concept satisfies IIS, ISDS and JD but not MC. To see why MC does not hold, consider the two reductions of $G$ whose sets of strategy profiles are $\left\lbrace U, D \right\rbrace \times \left\lbrace R \right\rbrace$ and $\left\lbrace D \right\rbrace \times \left\lbrace L, R \right\rbrace$. The strategy profile $(D,R)$ is a solution to both of these reductions and yet it is not a solution to $G$, thus violating MC.

Now consider the solution concept defined, for all $H \in \Gamma$, by:
\begin{equation*}
	\varphi(H) = \varphi^{\mathsf{NE}}(H) \cup \left\lbrace s \in S(H): s \sim_i t  \text{ for some  } t\in \varphi^{\mathsf{NE}}(H) \text{ and all } i\in I \right\rbrace,
\end{equation*}
where $S(H)$ is the set of strategy profiles of $H$. This solution concept satisfies MC, ISDS and JD but not IIS. To see why IIS does not hold, notice that $(D,L)$ is a solution to $G$ but not to the reduction that has $\left\lbrace U, D \right\rbrace \times \left\lbrace L \right\rbrace$ as its set of strategy profiles. \qed
\end{example}

\subsection{One-player games}
In one-player games, i.e., games in which $n = 1$, the axiomatization in Theorem \ref{thm:IFF} can be simplified by dropping the MC axiom. Specifically, the inclusion $\varphi^{\mathsf{NE}}(G) \subseteq \varphi(G)$ in part (b) of Lemma \ref{lemma:SolNash} becomes an immediate consequence of JD alone, thus making MC redundant.

As the following lemma shows, it is possible to drop both IIS and MC for solution concepts defined on a \textit{strictly closed} (instead of d-closed) class of one-player games. A non-empty class of games is said to be strictly closed if it contains all the strict reductions of every game contained in the class itself.

\begin{lemma}
Let $\Gamma$ be a strictly closed class of one-player games and let $\varphi$ be a solution concept on $\Gamma$.
\begin{itemize}
	\item[(a)] If $\varphi$ satisfies ISDS, then $\varphi(G) \subseteq \varphi^{\mathsf{NE}}(G)$ for all $G$ in $\Gamma$.
	\item[(b)] If $\varphi$ satisfies JD, then $\varphi^{\mathsf{NE}}(G) \subseteq \varphi(G)$ for all $G$ in $\Gamma$.
	\item[(c)] $\varphi$ satisfies ISDS and JD if and only if $\varphi (G) = \varphi^{\mathsf{NE}} (G)$ for all $G$ in $\Gamma$.
\end{itemize}
\end{lemma}
\begin{proof}
To show part (a), suppose by way of contradiction that $s \in \varphi(G)$ and $s \notin \varphi^{\mathsf{NE}}(G)$ for some strategy $s$ in a one-player game $G \in \Gamma$, with $\Gamma$ strictly closed. Then, there must exist a strategy $t$ in game $G$ such that $t \succ s$. Now take the reduction $G'$ of $G$ having the strategy set $S' = S \setminus \{s\}$, where $S$ is the strategy set of $G$. Clearly, $G'$ is a strict reduction of $G$, hence $G'$ is in $\Gamma$. Hence, by ISDS we have $\varphi (G) = \varphi (G')$, which contradicts the assumption that $s \in \varphi(G)$.

Parts (b) and (c) are immediate.
\end{proof}

\subsection{Domain}
The domain of the axiomatization in Theorem \ref{thm:IFF} is any d-closed class of games. For example, the class of all games is d-closed, as is the class of all finite games. By contrast, a class consisting only of one game is not necessarily d-closed. More importantly, a class consisting only of games with convex strategy sets may not be d-closed, the reason being that a doubleton strategy set is not convex. A case in point is the class of games having compact, convex strategy spaces and continuous payoff functions that are quasi-concave in each player's own strategy.\footnote{An axiomatization of Nash equilibrium for the games in this class with concave payoff functions, as well as for the class of mixed extensions of finite games, is provided by \cite{norde1996selection}.} Another important example of a class that is not d-closed is the class of mixed extensions of finite games. If a solution concept $\varphi$ satisfies the four axioms in Theorem \ref{thm:IFF} and is defined on the class of mixed extensions of finite games, then it may not select the mixed strategy Nash equilibria of a given finite game $G$. In order for $\varphi$ to select all the mixed strategy Nash equilibria of $G$, it must be defined on a d-closed class containing the mixed extension of $G$.

The axiomatic characterization of Nash equilibrium in Section \ref{sec:AxChar} can be extended to the class $\Gamma^*$ of all games having at least one Nash equilibrium,\footnote{As \citet[p. 85]{peleg1996minimality} point out, $\Gamma^*$ ``is not a well-described class; it does not have an independent description.'' By contrast, a d-closed class of games is well described in that its definition does not depend on any solution concept.} which is not d-closed. Specifically, it is easy to verify that Lemma \ref{lemma:SolNash} and Theorem \ref{thm:IFF} hold true if one substitutes $\Gamma^*$ for a generic d-closed class of games $\Gamma$.

\subsection{Other axioms from the existing literature}\label{sub:OtherAxioms}
One of the principal axioms in the early axiomatic literature on non-cooperative games is \textit{consistency} (CONS), introduced by \cite{peleg-tijs1996consistency}. CONS formalizes the following principle. Let $s = (s_1, \dots, s_n)$ be a solution to a game $G$. Suppose that $k$ out of $n$ players commit to play according to $s$ and leave the game. Then, the restriction of $s$ to the remaining $n - k$ players must be a solution to the reduced game formed with these remaining players alone. CONS is analogous to IIS, in that both axioms constrain the solution sets of ``smaller'' games obtained from a ``larger'' one. Nonetheless, IIS and CONS are logically independent. \citet[Example 2]{ray2000consistency} proves that CONS does not imply IIS. The following example shows that IIS does not imply CONS. 
\begin{example}
Consider the solution concept $\varphi$ defined, on the class of all finite games, by $\varphi(G) =\varphi^{\mathsf{NE}} (G)$ if the game $G$ has an even number of players, and $\varphi(G) = \emptyset$ otherwise. It is easy to verify that $\varphi$ satisfies IIS but not CONS. \qed
\end{example}

\textit{Converse consistency} (COCONS) and \textit{converse independence of irrelevant strategies} (CIIS) were introduced by \cite{peleg-tijs1996consistency} and \cite{forgo2015axiomsNash}, respectively. Just like MC in this paper, both COCONS and CIIS prescribe how to extend solutions from ``smaller'' to ``larger'' games. Informally, COCONS holds when the following is true. Take a game $G$ with at least two players and let $s$ be a strategy profile in it. Next consider all reduced games that one can form by removing some players from $G$. If, for any subgroup of players $I'$, the restriction of $s$ to the players in $I'$ is a solution to the reduced game whose set of players is $I'$, then $s$ is a solution of the original game $G$. The following example shows that COCONS and MC are logically independent.
\begin{example}
Suppose $\Gamma$ is the smallest class containing the game $G$ from Example \ref{ex:LogigIndep2}, all the reductions of $G$, and all the one-player games that one can form out of $G$. Consider the solution concept $\varphi$ that assigns to each one-player game in $\Gamma$ its entire set of strategy profiles, and to each two-player game in $\Gamma$ its set of Nash equilibria. It is easy to check that $\varphi$ satisfies MC but not COCONS. Conversely, let $\varphi'$ be the solution concept that assigns the empty set to each one-player game in $\Gamma$, and the set of strong Nash equilibria to each two-player game in $\Gamma$. One can check that $\varphi'$ satisfies COCONS but not MC. \qed
\end{example}

To conclude, the CIIS axiom posits that, if a strategy profile $s$ in a game $G$ (that has at least three strategy profiles) is a solution to all the proper reductions of $G$ to which $s$ itself belongs, then $s$ is a solution to $G$. Although MC implies CIIS, the example below proves that the converse is false.

\begin{example}
Let $G'$ be a two-player game in which preferences are induced by the payoff matrix below.

\begin{figure}[ht]
\centering
\begin{game}{3}{2}
	& $L$    & $R$ \\
$U$      & $2,1$  & $0,0$ \\
$C$      & $0,0$  & $1,2$ \\
$D$      & $2,1$  & $0,0$ 
\end{game}
\end{figure}

Suppose $\Gamma$ is the smallest class containing $G'$ and all its reductions. For all games $H$ in $\Gamma$, let $S_i(H)$ be player $i$'s set of strategies in $H$. Consider the solution concept $\varphi$ on $\Gamma$ defined as follows. If $S_2(H) = \{R\}$  and $\vert S_1(H) \times S_2(H)\vert = 2$, then $\varphi(H) = \emptyset$. Otherwise, we have:
\begin{equation*}
\varphi(H) = \{s \in \varphi^{\mathsf{NE}}(H): \not\exists t \in \varphi^{\mathsf{NE}}(H) \text{ for which } t \succ_1 s\}.
\end{equation*}

One can verify that $\varphi$ satisfies CIIS but not MC.
\qed
\end{example}



\bibliographystyle{plainnat}
\bibliography{biblio}

\end{document}